\documentclass[11pt]{article}
\usepackage[margin=1in]{geometry}
\usepackage[utf8]{inputenc}
\usepackage{microtype}
\usepackage{amsmath,amsthm,amssymb, color, tikz, mathtools}
\usepackage{caption}

\usepackage{algorithm}
\usepackage{algpseudocode}
\usepackage[hidelinks]{hyperref}
\usepackage{thm-restate}

\newcommand{\set}[1]{\left\{ #1 \right\}}
\newcommand{\IND}{\mathsf{IND}}

\newcommand{\zone}{\set{0,1}}
\newcommand{\cbra}[1]{\left\{#1\right\}}
\newcommand{\rbra}[1]{\left(#1\right)}

\newtheorem{theorem}{Theorem}[section]
\newtheorem{corollary}[theorem]{Corollary}

\newtheorem{lemma}[theorem]{Lemma}

\newtheorem{defi}[theorem]{Definition}

\newtheorem{proposition}[theorem]{Proposition}

\newcommand{\ket}[1]{\ensuremath{\left|#1\right\rangle}}
\newcommand{\bra}[1]{\ensuremath{\left\langle#1\right|}}
\newcommand{\braket}[2]{\ensuremath{\left\langle#1\middle|#2\right\rangle}}
\newcommand{\norm}[1]{\ensuremath{\left\|#1\right\|}}

\newcommand{\sab}{\mathsf{sab}} 
\newcommand{\ADV}{\mathsf{ADV}}
\newcommand{\ind}{\mathsf{ind}}
\newcommand{\bin}{\mathsf{bin}}
\newcommand{\weak}{\mathsf{weak}}
\newcommand{\strong}{\mathsf{str}}
\newcommand{\bs}{\mathsf{bs}}
\newcommand{\fbs}{\mathsf{fbs}}
\newcommand{\D}{\mathsf{D}}
\newcommand{\Q}{\mathsf{Q}}
\newcommand{\QD}{\mathsf{QD}}
\newcommand{\QS}{\mathsf{QS}}
\newcommand{\R}{\mathsf{R}}
\newcommand{\RS}{\mathsf{RS}}

\newcommand{\sabStar}[2]{[#1,#2,*]}
\newcommand{\sabDagger}[2]{[#1,#2,\dagger]}
\newcommand{\tr}{\mathsf{tr}}
\newcommand{\Tr}{\mathrm{Tr}}
\newcommand{\domainSet}{\Sigma}
\newcommand{\rangeSet}{\Pi}

\title{Quantum Sabotage Complexity}
\author{
    Arjan Cornelissen\thanks{IRIF - CNRS, Paris, France {\tt ajcornelissen@outlook.com}}
    \and
    Nikhil S.~Mande\thanks{University of Liverpool, UK {\tt mande@liverpool.ac.uk}}
    \and
    Subhasree Patro\thanks{Eindhoven University of Technology (TU/e) \& Centrum Wiskunde en Informatica (QuSoft), The Netherlands {\tt patrofied@gmail.com}. This work was partially done while the author was affiliated to Utrecht University.}
}

\begin{document}

    \maketitle

    \begin{abstract}
        Given a Boolean function $f : \zone^n \to \zone$, the goal in the usual query model is to compute $f$ on an unknown input $x \in \zone^n$ while minimizing the number of queries to $x$. One can also consider a ``distinguishing'' problem denoted by $f_\sab$: given an input $x \in f^{-1}(0)$ and an input $y \in f^{-1}(1)$, either all differing locations are replaced by a $*$, or all differing locations are replaced by $\dagger$, and an algorithm's goal is to identify which of these is the case while minimizing the number of queries.

        Ben-David and Kothari [ToC'18] introduced the notion of randomized sabotage complexity of a Boolean function to be the zero-error randomized query complexity of $f_\sab$.
        A natural follow-up question is to understand $\Q(f_\sab)$, the quantum query complexity of $f_\sab$. In this paper, we initiate a systematic study of this. The following are our main results for all Boolean functions $f : \zone^n \to \zone$.
        \begin{itemize}
            \item If we have additional query access to $x$ and $y$, then $\Q(f_\sab) = O(\min\{\Q(f),\sqrt{n}\})$.
            \item If an algorithm is also required to output a differing index of a 0-input and a 1-input, then $\Q(f_\sab) = O(\min\cbra{\Q(f)^{1.5}, \sqrt{n}})$.
            \item $\Q(f_\sab) = \Omega(\sqrt{\fbs(f)})$, where $\fbs(f)$ denotes the fractional block sensitivity of $f$. By known results, along with the results in the previous bullets, this implies that $\Q(f_\sab)$ is polynomially related to $\Q(f)$.
            \item The bound above is easily seen to be tight for standard functions such as And, Or, Majority and Parity. We show that when $f$ is the Indexing function, $\Q(f_\sab) = \Theta(\fbs(f))$, ruling out the possibility that $\Q(f_\sab) = \Theta(\sqrt{\fbs(f)})$ for all $f$.
        \end{itemize}
    \end{abstract}

    \section{Introduction}

    Given a Boolean function $f : \zone^n \to \zone$, the goal in the standard query complexity model is to compute $f(x)$ on an unknown input $x \in \zone^n$ using as few queries to $x$ as possible. One can also consider the following distinguishing problem: given $x \in f^{-1}(0)$ and $y \in f^{-1}(1)$, output an index $i \in [n]$ such that $x_i \neq y_i$. This task can be formulated as follows: consider an arbitrary $x \in f^{-1}(0)$, an arbitrary $y \in f^{-1}(1)$, and then either all indices where $x$ and $y$ differ are replaced by the symbol $*$, or all such indices are replaced by the symbol $\dagger$. The goal of an algorithm is to identify which of these is the case, with query access to this ``sabotaged'' input. A formal description of this task is given below.

    Let $f: D \rightarrow \{0,1\}$ with $D \subseteq \{0,1\}^n$ be a (partial) Boolean function.
    Let $D_0 = f^{-1}(0)$ and $D_1 = f^{-1}(1)$. For any pair $(x,y) \in D_0 \times D_1$,
    define $\sabStar{x}{y} \in \cbra{0, 1, *}^n$ to be
    \begin{align*}
        \sabStar{x}{y}_i = \begin{cases}
            x_i, & \textnormal{ if } x_i = y_i, \\
            *, & \textnormal{ otherwise.}
        \end{cases}
    \end{align*}
    Similarly, for any pair $(x,y) \in D_0 \times D_1$, define $\sabDagger{x}{y} \in \cbra{0, 1, \dagger}^n$ to be
    \begin{align*}
        \sabDagger{x}{y}_i = \begin{cases}
            x_i, & \textnormal{ if } x_i = y_i, \\
            \dagger, & \textnormal{ otherwise.}
        \end{cases}
    \end{align*}
    Let $S_{*}=\{\sabStar{x}{y} \mid (x,y) \in D_0 \times D_1 \}$ and $S_{\dagger}=\{\sabDagger{x}{y} \mid (x,y) \in D_0 \times D_1 \}$. That is, $S_*$ is the set of $*$-sabotaged inputs for $f$ and $S_\dagger$ is the set of $\dagger$-sabotaged inputs for $f$.
    Finally, let $f_{\sab}: S_{*} \cup S_{\dagger} \rightarrow \{0,1\}$ be the function defined by
    \begin{align*}
        f_{\sab}(z) = \begin{cases}
            0, & z \in S_*, \\
            1, & z \in S_\dagger.
        \end{cases}
    \end{align*}
    That is, $f_\sab$ takes as input a sabotaged input to $f$ and identifies if the input is $*$-sabotaged or if it is $\dagger$-sabotaged.

    Ben-David and Kothari~\cite{BK18} introduced the notion of \emph{randomized sabotage complexity} of a Boolean function $f$, defined to be $\RS(f) := \R_0(f_\sab)$, where $\R_0(\cdot)$ denotes randomized zero-error query complexity. It is not hard to see that $\RS(f) = O(\R(f))$ (where $\R(\cdot)$ denotes randomized bounded-error query complexity); this is because a randomized algorithm that succeeds with high probability on both a 0-input $x$ and a $1$-input $y$ must, with high probability, query an index where $x$ and $y$ differ. Ben-David and Kothari also showed that randomized sabotage complexity admits nice composition properties. It is still open whether $\RS(f) = \Theta(\R(f))$ for all total Boolean functions $f$. If true, this would imply that randomized query complexity admits a perfect composition theorem, a goal towards which a lot of research has been done~\cite{AGJKLMSS17, GJPW18, GLSS19, BB20, BDGHMT20, BBGM22, CKMPSS23, San24}. This motivates the study of randomized sabotage complexity.

    In the same paper, they mentioned that one could define $\QS(f) := \Q(f_\sab)$ (here, $\Q(\cdot)$ denotes bounded-error quantum query complexity), but they were unable to show that it lower bounds $\Q(f)$. In a subsequent work~\cite{BK19}, they defined a quantum analog, denoted $\QD(f)$, and called it \emph{quantum distinguishing complexity}. $\QD(f)$ is the minimum number of quantum queries to the input $x \in D$ such that the final states corresponding to 0-inputs and 1-inputs are far from each other. Analogous to their earlier result, they were able to show that $\QD(f) = O(\Q(f))$ for all total $f$. Additionally, using $\QD(f)$ as an intermediate measure, they were able to show a (then) state-of-the-art 5th-power relationship between zero-error quantum query complexity and bounded-error quantum query complexity: $\Q_0(f) = \widetilde{O}(\Q(f)^5)$ for all total $f$. We note here that a 4th-power relationship was subsequently shown between $\D(f)$ and $\Q(f)$~\cite{ABKRT21}, also implying $\Q_0(f) = O(\Q(f)^4)$ for all total $f$. The proof of this relied on Huang's celebrated sensitivity theorem~\cite{Huang19}.

    \subsection{Our results}

    In this paper, we initiate a systematic study of the natural quantum analog of randomized sabotage complexity alluded to in the previous paragraph, which we call \emph{quantum sabotage complexity}, denoted by $\QS(f) := \Q(f_\sab)$. Slightly more formally, we consider the following variants:
    \begin{itemize}
        \item We consider two input models. In the weak input model, the oracle simply has query access to an input in $z \in S_* \cup S_\dagger$. In the strong input model, the oracle additionally has access to the original inputs $x \in D_0$ and $y \in D_1$ that yield the corresponding input in $S_* \cup S_\dagger$. The model under consideration will be clear by adding either $\weak$ or $\strong$ as a subscript to $\QS$.
        \item We also consider two different output models: one where an algorithm is only required to output whether the input was in $S_*$ or in $S_\dagger$, and a stronger version where an algorithm is required to output an index $i \in [n]$ with $z_i \in \cbra{*, \dagger}$. The model under consideration will be clear by adding either no superscript or $\ind$ as a superscript to $\QS$.
    \end{itemize}
    As an example, $\QS_\weak^\ind(f)$ denotes the quantum sabotage complexity of $f$ under the weak input model, and in the output model where an algorithm needs to output a differing index.

    One can also consider these nuances in defining the input and output models in the randomized setting. However, one can easily show that they are all equivalent for randomized algorithms. We refer the reader to Section~\ref{sec:QuantumSabotageComplexity} for a proof of this, and for a formal description of these models.

    An immediate upper bound on $\QS(f)$, for any (partial) Boolean function $f$, follows directly from Grover's algorithm~\cite{Gro96}. Indeed, for any sabotaged input, we know that at least one of the input symbols must be either a $*$ or $\dagger$. Thus, we can simply use the unstructured search algorithm by Grover to find (the position of) such an element in $O(\sqrt{n})$ queries. This immediately tells us that for Boolean functions where $\Q(f) = \omega(\sqrt{n})$, $\QS(f)$ is significantly smaller than $\Q(f)$.

    As mentioned earlier, Ben-David and Kothari left open the question of whether $\QS(f) = O(\Q(f))$, the quantum analog of randomized sabotage complexity being at most randomized query complexity.
    We first observe that in the strong input model, this holds true.

    \begin{restatable}{lemma}{qsvsq}
        \label{lem:qsvsq}
        Let $n$ be a positive integer, let $D \subseteq \zone^n$. Let $f : D \to \zone$ be a (partial) Boolean function. Then,
        \[\QS_\strong(f) = O(\Q(f)).\]
    \end{restatable}

    The proof idea is simple: consider an input $z \in S_* \cup S_\dagger$ obtained by sabotaging $x \in f^{-1}(0)$ and $y \in f^{-1}(1)$. Run a quantum query algorithm for $f$ on input $z$, such that:
    \begin{itemize}
        \item Whenever a bit in $\zone$ is encountered, the algorithm proceeds as normal.
        \item Whenever a $*$ is encountered, query the corresponding bit in $x$.
        \item Whenever a $\dagger$ is encountered, query the corresponding bit in $y$.
    \end{itemize}
    The correctness follows from the following observation:  if $z \in S_*$, then the run of the algorithm is exactly that of the original algorithm on $x$, and if $z \in S_\dagger$, then the run is the same of the original algorithm on $y$.

    This procedure does not return a $*/\dagger$-index, and it is natural to ask if $\QS_\strong^\ind(f) = O(\Q(f))$ as well. While we are unable to show this, we make progress towards this by showing the following, which is our first main result.

    \begin{restatable}{theorem}{qsindvsq}
        \label{thm: main 1}
        Let $n$ be a positive integer, let $D \subseteq \zone^n$, and let $f : D \to \zone$ be a (partial) Boolean function. Then,
        \begin{align*}
            \QS_\strong^\ind(f) = O(\Q(f)^{1.5}).
        \end{align*}
    \end{restatable}

    Our result is actually slightly stronger than this; our proof shows that $\QS_\strong^\ind(f) = O(\QD(f)^{1.5})$. This implies Theorem~\ref{thm: main 1} using the observation of Ben-David and Kothari that $\QD(f) = O(\Q(f))$ for all (partial) $f$.

    In order to show Theorem~\ref{thm: main 1}, we take inspiration from the observation that randomized sabotage complexity is at most randomized query complexity. This is true because with high probability, a randomized query algorithm must spot a differing bit between any pair of inputs with different output values. However, this argument cannot immediately be ported to the quantum setting because quantum algorithms can make queries in superposition. We are able to do this, though, by stopping a quantum query algorithm for $f$ at a random time and measuring the index register. We note here that it is important that we have oracle access not only to a sabotaged input $z \in S_* \cup S_\dagger$, but also the inputs $(x, y) \in D_0 \times D_1$ that yielded the underlying sabotaged input. Using arguments reminiscent of the arguments in the hybrid method~\cite{BBBV97}, we are able to show that the success probability of this is only $1/\Q(f)$. Applying amplitude amplification to this process leads to an overhead of $\sqrt{\Q(f)}$, and yields Theorem~\ref{thm: main 1}.

    We remark here that this proof idea is reminiscent of the proof of~\cite[Theorem~1]{BK19}. However, there are some technical subtleties. At a high level, the main subtleties are the following: in the strong oracle model that we consider, it is easy to check whether or not a terminated run of a $\Q$ algorithm actually gives us a $*/\dagger$ index. This is not the case in~\cite[Proof of Theorem~1]{BK19}. This allows us to save upon a quadratic factor because we can do amplitude amplification.

    Our proof approach modifies a core observation by Ben-David and Kothari~\cite[Lemma~12]{BK19}. In their setting, they consider a $T$-query algorithm $\mathcal{A}$ that computes a function $f$ with high probability. Take two inputs $x$ and $y$ such that $f(x) \neq f(y)$, and let $B \subseteq [n]$ be the indices on which $x$ and $y$ differ. Suppose we run this algorithm on $x$, interrupt it right before the $t$th query, and then measure the query register. The probability that we measure an index $i \in B$, we denote by $p_t$. Then, \cite[Lemma~12]{BK19} shows that
    \[\sum_{t=1}^T p_t = \Omega\left(\frac{1}{T}\right).\]

    We show that by adding the probabilities that come from running the algorithm at input $y$, we can replace the lower bound of $\Omega(1/T)$ by a much stronger lower bound of $\Omega(1)$. We state this observation more formally in the following lemma.

    \begin{restatable}{lemma}{hybrid}
        \label{lem: hybrid new}
        Let $x \in \zone^n$ be an input and let $B \subseteq [n]$. Let $\mathcal{A}$ be a $T$-query quantum algorithm that accepts $x$ and rejects $x_B$ with high probability, or more generally produces output states that are a constant distance apart in trace distance for $x$ and $x_B$. Let $p_t$, resp.\ $p_t^B$, be the probability that, when $\mathcal{A}$ is run on $x$, resp.\ $x_B$, up until, but not including, the $t$-th query and then measured, it is found to be querying a position $i \in B$. Then,
        \begin{align*}
            \sum_{t = 1}^T (p_t + p_t^B) = \Omega(1).
        \end{align*}
    \end{restatable}

    We find this lemma independently interesting and are confident that it will find use in future research, given the use of such statements in showing quantum lower bounds via the adversary method, for example (see~\cite[Chapters 11-12]{Wol23} and the references therein). Note that this lemma is very amenable to our ``strong'' sabotage complexity setup: imagine running an algorithm simultaneously on inputs $x$ and $x_B$ that have different function values. Lemma~\ref{lem: hybrid new} says that on stopping at a random time in the algorithm, and choosing one of $x$ and $x_B$ at random, the probability of seeing an index in $B$ (i.e., a $*$-index or a $\dagger$-index) is a constant.
    Indeed, this lemma is a natural quantum generalization of the phenomenon that occurs in the randomized setting: a randomized algorithm that distinguishes $x$ and $x_B$ must read an index in $B$ with constant probability (on input either $x$ or $x_B$).

    We now discuss our remaining results. Given Theorem~\ref{thm: main 1}, it is natural to ask if $\QS_\strong^\ind(\cdot)$ is polynomially related to $\Q(\cdot)$. Using the positive-weighted adversary lower bound for quantum query complexity~\cite{Amb06}, we are able to show that $\QS_\strong(f) = \Omega(\sqrt{\fbs(f)})$ for all total $f$ (and hence the same lower bound holds for $\QS_\weak(f)$, and $\QS_\weak^\ind(f)$, and $\QS_\strong^\ind(f)$ as well).

    \begin{theorem}\label{thm: qsab bs lower bound}
        Let $f : \zone^n \to \zone$ be a Boolean function. Then
        \begin{align*}
            \QS_\strong(f) = \Omega(\sqrt{\fbs(f)}).
        \end{align*}
    \end{theorem}

    Fractional block sensitivity is further lower bounded by block sensitivity, which is known to be polynomially related to $\R(\cdot)$ and $\Q(\cdot)$~\cite{Nisan91, BBCMW01}. Using the best-known relationship of $\Q(f) = \widetilde{O}(\bs^3(f))$~\cite{ABK16}, Theorem~\ref{thm: main 1} and Theorem~\ref{thm: qsab bs lower bound} implies the following polynomial relationship between $\QS_\strong(f)$, $\QS_\strong^\ind(f)$ and $\Q(f)$ for all total Boolean functions $f$.
    \begin{align}
        \QS_\strong(f) = O(\Q(f)),& \qquad \Q(f) = \widetilde{O}(\QS_\strong(f)^6). \\
        \QS_\strong^\ind(f) = O(\Q(f)^{1.5}),& \qquad \Q(f) = \widetilde{O}(\QS_\strong^\ind(f)^6).
    \end{align}
    In the weakest input model, we have $\QS_{\weak}(f) = O(\QS_\weak^\ind(f)) = O(\R(f_\sab)) = O(\R(f)) = O(\Q(f)^4)$ (where the last inequality follows from~\cite{ABKRT21}). Thus, in the weakest input model for sabotage complexity, Theorem~\ref{thm: qsab bs lower bound} implies the following polynomial relationship with $\Q(f)$:
    \begin{align}
        \QS_\weak(f) = O(\Q(f)^{4}),& \qquad \Q(f) = \widetilde{O}(\QS_\weak(f)^6). \\
        \QS_\weak^\ind(f) = O(\Q(f)^{4}),& \qquad \Q(f) = \widetilde{O}(\QS_\weak^\ind(f)^6).
    \end{align}
    It would be interesting to find the correct polynomial relationships between all of these measures. In particular, we suspect that $\QS_\strong^\ind(f) = O(\Q(f))$ for all total $f$, but we are unable to show this.

    As mentioned in the discussion before Theorem~\ref{thm: main 1}, it is easy to show that $\QS_\strong^\ind(f) = O(\sqrt{n})$ for all $f : \zone^n \to \zone$. Thus, the lower bound in terms of block sensitivity given by Theorem~\ref{thm: qsab bs lower bound} is actually tight for standard functions like And, Or, Parity and Majority. Given this, it is natural to ask if $\QS_\strong^\ind(f) = \Theta(\sqrt{\fbs(f)})$ for all total Boolean $f$. We show that this is false, by showing that for the Indexing function $\IND_n : \zone^{n + 2^n} \to \zone$ defined by $\IND_n(x, y) = y_{\bin(x)}$ (where $\bin(x)$ denotes the integer in $[2^n]$ with the binary representation $x$), $\QS_\strong^\ind(\IND_n) = \Theta(\fbs(\IND_n)) = \Theta(n)$.

    \begin{theorem}\label{thm: ind lower bound}
        Let $n$ be a positive integer. Then,
        \begin{align*}
            \QS_\strong^\ind(\IND_n) = \Theta(n).
        \end{align*}
    \end{theorem}
    In order to show this, we use a variation of Ambainis's basic adversary method~\cite[Theorem~5.1]{Amb02}, also presented in the same paper~\cite[Theorem~6.1]{Amb02} (see Lemma~\ref{lem: ambainis v2}).

    Finally, we summarize the relations we proved in Figure~\ref{fig:relations}.

    \begin{figure}[!ht]
        \centering
        \begin{tikzpicture}[scale=1.2,measure/.style={draw, rounded corners=.5em, blue}]
            \node[measure] (fbs) at (0,0) {$\sqrt{\fbs}$};
            \node[measure] (QSs) at (1,1) {$\QS_{\strong}$};
            \node[measure] (QSsi) at (0,2) {$\QS_{\strong}^{\ind}$};
            \node[measure] (QSw) at (2,2) {$\QS_{\weak}$};
            \node[measure] (QSwi) at (1,3) {$\QS_{\weak}^{\ind}$};
            \node (QD32) at (-1,3) {$\QD^{3/2}$};
            \draw (fbs) to node[rotate=0,right=.2em] {\small Theorem~\ref{thm: qsab bs lower bound}} (QSs);
            \draw[red] (QSs) to (QSsi) to (QSwi) to (1.75,3.75) node[above right] {\color{black}$\sqrt{n}$};
            \draw[red] (QSs) to (QSw) to (QSwi) to (.25,3.75) node[above left] {\color{black}$\RS$};
            \draw (fbs) to node[rotate=0,left=.2em] {\small Lemma~\ref{lem:fbsrs}} (-.75,.75) node[above left] {$\sqrt{\RS}$};
            \draw (QSsi) to node[rotate=0,left=.2em] {\small Theorem~\ref{thm:qsvsqd}} (QD32) to node[rotate=0,left=.4em] {\small\cite[Proposition~6]{BK19}} (-1.75,3.75) node[above left] {$\Q^{3/2}$};
            \draw (QSs) to node[rotate=0,below=.35em] {\small Lemma~\ref{lem:qsvsq}} (3.25,1.5) node[ right] {$\Q$};
        \end{tikzpicture}
        \caption{Overview of the relations proved in this work. If nodes $A$ and $B$ are connected, and $A$ is below $B$, then $A = O(B)$. All the red edges are reasonably straightforward inclusions, and they are proved in Proposition~\ref{prop:relations-quantum-measures}.}
        \label{fig:relations}
    \end{figure}
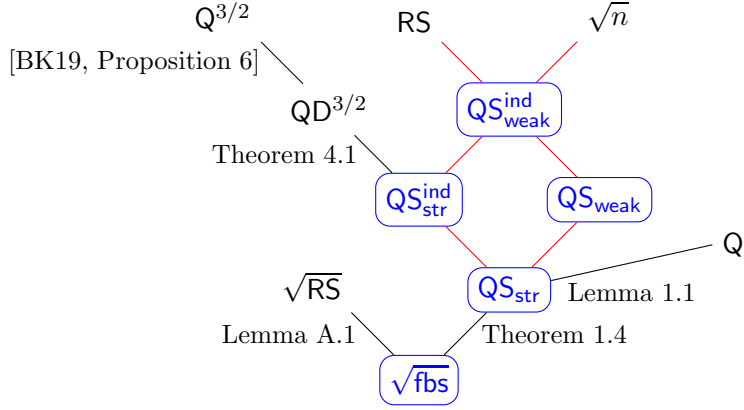

    \section{Preliminaries}

    All logarithms in this paper are taken base 2 unless mentioned otherwise. For a positive integer $n$, we use the notation $[n]$ to denote the set $\cbra{1, 2, \dots, n}$ and use $[n]_0$ to denote $\cbra{0, 1, \dots, n-1}$. Let $v \in \mathbb{C}^d$, then $\norm{v}_2=\sqrt{\sum_{i=1}^{d} |v_i|^2}$. Let $A, B$ be square matrices in $\mathbb{C}^{d \times d}$. We use $A^{\dagger}$ to denote the conjugate transpose of matrix $A$. The operator norm of a matrix $A$, denoted by $\norm{A}$, is the largest singular value of $A$, i.e., $\norm{A}=\max_{v : \norm{v}_2=1} \norm{Av}_2$.
    The trace distance between two matrices $A,B$, denoted by $\norm{A-B}_{\mathsf{tr}}=\frac{1}{2}\norm{A-B}_1$ where $\norm{A}_1=\Tr(\sqrt{A^{\dagger}A})$. For two $d \times d$ matrices $A, B$, $A \circ B$ denotes the Hadamard, or entry-wise, product of $A$ and $B$.

    We refer the reader to \cite[Chapter~1]{Wol23} for the relevant basics of quantum computing.

    Let $D \subseteq \zone^n$, let $R$ be a finite set, and let $f \subseteq D \times R$ be a relation. A quantum query algorithm $\mathcal{A}$ for $f$ begins in a fixed initial state $\ket{\psi_0}$ in a finite-dimensional Hilbert space, applies a sequence of unitaries $U_0, O_x, U_1, O_x, \dots, U_T$, and performs a measurement. Here, the initial state $\ket{\psi_0}$ and the unitaries $U_0, U_1, \dots, U_T$ are independent of the input. The unitary $O_x$ represents the ``query'' operation, and does the following for each basis state:
    it maps $\ket{i}\ket{b}$ to $\ket{i}\ket{b + x_i \mod 2}$ for all $i \in [n]$.
    The algorithm then performs a measurement and outputs the observed value.
    We say that $\mathcal{A}$ is a bounded-error quantum query algorithm computing $f$ if for all $x \in D$ the probability of outputting $r$ such that $(x, r) \in f$ is at least $2/3$. The (bounded-error) \emph{quantum query complexity of $f$}, denoted by $\Q(f)$, is the least number of queries required for a quantum query algorithm to compute $f$ with error probability at most $1/3$. We use $\R(f)$ to denote the \emph{randomized query complexity} of $f$, which is the worst-case cost (number of queries) of the best randomized algorithm that computes $f$ with error probability at most $1/3$ on all inputs.

    We recall some known complexity measures.

    \begin{defi}[Block sensitivity]
        Let $n$ be a positive integer and $f : \zone^n \to \zone$ be a Boolean function. For any $x \in \zone^n$, a block $B \subseteq [n]$ is said to be sensitive on an input $x \in \zone^n$ if $f(x) \neq f(x \oplus B)$, where $x \oplus B$ (or $x_B$) denotes the string obtained by taking $x$ and flipping all bits in $B$. The \emph{block sensitivity of $f$ on $x$}, denoted $\bs(f, x)$, is the maximum number of pairwise disjoint blocks that are sensitive on $x$. The \emph{block sensitivity of $f$}, denoted by $\bs(f)$, is $\max_{x \in \zone^n} \bs(f,x)$.
    \end{defi}

    The \emph{fractional block sensitivity} of $f$ on $x$, denoted $\fbs(f, x)$, is the optimum value of the linear program below. We refer the reader to~\cite{GSS16, KT16} for a formal treatment of fractional block sensitivity and related measures, and we simply state its definition here.

    \begin{defi}[Fractional block sensitivity]
        \label{def:fbs}
        Let $n$ be a positive integer and $f : \zone^n \to \zone$ be a Boolean function. Let $x \in \zone^n$, and let $X,Y \subseteq \zone^n$ be the set of all inputs $z$ whose function values satisfy $f(x) = f(z)$, resp.\ $f(x) \neq f(z)$. The fractional block sensitivity of $f$ on $x$, denoted by $\fbs(f,x)$, is the optimal value of the following optimization program.
        \begin{align*}
            \max \quad& \sum_{y \in Y} w_y, \\
            \textnormal{subject to}\quad & \sum_{\substack{y \in Y \\ x_j \neq y_j}} w_y \leq 1, & \forall j \in [n], \\
            & w_y \geq 0, & \forall y \in Y.
        \end{align*}
        The fractional block sensitivity of $f$, denoted by $\fbs(f)$, is $\max_{x \in \zone^n} \fbs(f,x)$.\footnote{The block sensitivity of $f$ on $x$ is captured by the integral version of this linear program, where the variable $w_y \in \{0,1\}$ enforces that the blocks must be disjoint.}
    \end{defi}

    We also state the non-negative weight adversary bound. It appears in many forms in the literature. We refer to the form mentioned in \cite{HLS07} which is equivalent to the following definition.

    \begin{defi}[Non-negative weight adversary bound]
        \label{def:adv}
        Let $n$ be a positive integer, $\domainSet$ and $\rangeSet$ finite sets, $D \subseteq (\domainSet)^n$, and $f : D \to \rangeSet$. The adversary bound is the following optimization program.
        \begin{align*}
            \max\quad & \norm{\Gamma} \\
            \text{s.t.}\quad & \norm{\Gamma \circ \Delta_j} \leq 1, & \forall j \in [n], \\
            & \Gamma[x,y] = 0, & \text{if } f(x) = f(y).
        \end{align*}
        Here, the optimization is over all symmetric, entry-wise non-negative adversary matrices $\Gamma \in \mathbb{R}^{D \times D}$. The matrix $\Delta_j \in \zone^{D \times D}$ has entries $\Delta_j[x,y] = 1$ if and only if $x_j \neq y_j$. The optimal value of this optimization program is denoted by $\ADV^+(f)$.
    \end{defi}

    Sometimes, the optimal value of the non-negative weight adversary bound is also written as $\ADV(f)$. However, one can also consider the general adversary bound, in which the entries of the matrix $\Gamma$ are not constrained to be non-negative. To clearly differentiate between the optimal values of these optimization programs, we distinguish between them by explicitly writing $\ADV^+(f)$ and $\ADV^{\pm}(f)$.

    \begin{defi}[Quantum distinguishing complexity]
        Let $n$ be a positive integer. The quantum distinguishing complexity of a (partial) Boolean function $f : D \rightarrow \zone$ (where $D \subseteq \zone^n$), denoted by $\QD(f)$, is the smallest integer $k$ such that there exists a $k$-query algorithm that on input $x \in D$ outputs a quantum state $\rho_x$ such that
        \begin{align*}
            \forall x,y \in D, \quad \norm{\rho_x -\rho_y}_{\tr} \geq 1/6,
        \end{align*}
        whenever $f(x)\neq f(y)$.
    \end{defi}

    The non-negative weight adversary bound is known to be a lower bound to the quantum distinguishing complexity, which in turn is a lower bound on the quantum query complexity, by \cite[Proposition~6]{BK19}. We state it below.

    \begin{theorem}[\cite{BK19}]\label{thm:advplus lower bound}
        Let $n$ be a positive integer, $D \subseteq \zone^n$, and $f : D \to \zone$. Then,
        \[\ADV^+(f) = O(\QD(f)) = O(\Q(f))).\]
    \end{theorem}

    The relation $\ADV^+(f) = O(\Q(f))$ holds in the non-Boolean case as well, which follows directly from the definition and known results about the non-negative weight adversary bound, as can be found in \cite{LMRSS11}, for instance.

    \section{Sabotage complexity}
    \label{sec:QuantumSabotageComplexity}

    In this section we first define sabotage variants of a Boolean function $f$ that are convenient to work with. Specifically, these variants are useful because they enable us to work with the usual quantum query complexity model in the quantum setting, allowing us to use known results in this setting. After this, we analyze some basic properties of randomized and quantum sabotage complexities.

    \subsection{Formal setup of sabotage complexity}

    We start by formally defining the sabotage function of $f$.

    \begin{defi}[Sabotage functions and relations]
        \label{def:fsab}
        Let $n$ be a positive integer, $D \subseteq \zone^n$. Let $f : D \to \zone$ be a (partial) Boolean function. For any input pair $x \in f^{-1}(0)$ and $y \in f^{-1}(1)$, we define $\sabStar{x}{y},\sabDagger{x}{y} \in \{0,1,*,\dagger\}^n$ by
        \[\sabStar{x}{y}_j = \begin{cases}
            x_j, & \text{if } x_j = y_j, \\
            *, & \text{otherwise},
        \end{cases} \qquad \text{and} \qquad \sabDagger{x}{y}_j = \begin{cases}
            x_j, & \text{if } x_j = y_j, \\
            \dagger, & \text{otherwise}.
        \end{cases}\]
        We let $S_* = \{\sabStar{x}{y} : x \in f^{-1}(0), y \in f^{-1}(1)\}$, $S_{\dagger} = \{\sabDagger{x}{y} : x \in f^{-1}(0), y \in f^{-1}(1)\}$, and we let $D_{\sab} = S_* \cup S_{\dagger} \subseteq \{0,1,*,\dagger\}^n$.
        \begin{itemize}
            \item The \emph{sabotage function of $f$} is defined as $f_{\sab} : D_{\sab} \to \zone$, where $f_\sab(z) = 1$ iff $z \in S_\dagger$.
            \item The \emph{sabotage relation of $f$} is defined as $f^{\ind}_{\sab} \subseteq D_{\sab} \times [n]$, where $(z, j) \in f_\sab^\ind$ iff $z_j \in \{*,\dagger\}$.
        \end{itemize}

        For every $x \in f^{-1}(0)$ and $y \in f^{-1}(1)$, we let $(x,y,*)$ denote $((x_j,y_j,z_j))_{j=1}^n$, where $z = \sabStar{x}{y}$. Similarly, for every $x \in f^{-1}(0)$ and $y \in f^{-1}(1)$, we let $(x,y,\dagger)$ denote $((x_j,y_j,z_j))_{j=1}^n$, where $z = \sabDagger{x}{y}$. For $b \in \cbra{*, \dagger}$, let $S^{\strong}_b = \{(x,y,b) : x \in f^{-1}(0), y \in f^{-1}(1)\}$, and $D^{\strong}_{\sab} = S^{\strong}_* \cup S^{\strong}_{\dagger}$.
        \begin{itemize}
            \item We define the \emph{strong sabotage function of $f$} as $f_{\sab}^{\strong} : D^{\strong}_{\sab} \to \zone$, where $f_\sab^\strong(w) = 1$ iff $w \in S_\dagger^\strong$.
            \item We define the \emph{strong sabotage relation of $f$} as $f_{\sab}^{\strong, \ind} \subseteq D^{\strong}_{\sab} \times [n]$, where $(w, j) \in f_\sab^{\strong, \ind}$ iff $z_j \in \cbra{*, \dagger}$ where $w = ((x_j,y_j,z_j))_{j=1}^n$.
        \end{itemize}
    \end{defi}

    If we want to compute $f_{\sab}$, we need to consider how we are given access to the input of $f_{\sab}$. To that end, we consider two input models, the weak and the strong input model. Both can be viewed as having regular query access to the inputs of the function $f_{\sab}$ and $f_{\sab}^{\strong}$, respectively.

    \begin{defi}[Weak sabotage input model]
        \label{def:WeakInputModel}
        Let $n$ be a positive integer, $D \subseteq \zone^n$, and $f : D \to \zone$ be a (partial) Boolean function. Let $D_{\sab}$ be as in Definition~\ref{def:fsab}. In the \emph{weak sabotage input model}, on some input $z \in D_{\sab}$, we are given access to an oracle $O^{\weak}_z$ that when queried the $j$th position returns $z_j$. In the quantum setting, this means that the oracle performs the mapping
        \[
        O^{\weak}_z : \ket{j}\ket{b} \mapsto \ket{j}\ket{(b + z_j) \mod 4}, \quad \forall b \in [4]_0, \forall j \in [n],
        \]
        where $*$ is identified with $2$ and $\dagger$ is identified with $3$.
    \end{defi}

    We also consider a stronger model.

    \begin{defi}[Strong sabotage input model]
        \label{def:StrongInputModel}
        Let $n$ be a positive integer, $D \subseteq \zone^n$, and $f : D \to \zone$ be a (partial) Boolean function. Let $D^{\strong}_{\sab}$ be as in Definition~\ref{def:fsab}. In the \emph{strong sabotage input model}, on some input $w = ((x_j,y_j,z_j))_{j=1}^n$,
        we are given access to an oracle $O^{\strong}_w$ that when queried the $j$th position returns the tuple $(x_j, y_j, z_j)$.
        In the quantum setting, this means that the oracle performs the mapping
        \[
        O_w^{\strong} : \ket{j}\ket{b_x}\ket{b_y}\ket{b_z} \mapsto \ket{j}\ket{b_x \oplus x_j}\ket{b_y \oplus y_j}\ket{(b_z + z_j) \mod 4},
        \]
        for all $b_x,b_y \in \zone$, $b_z \in [4]_0$ and for all $j \in [n]$. As in the previous definition, $*$ is identified with $2$ and $\dagger$ is identified with $3$.
    \end{defi}

    Note that in the stronger model, we are implicitly also given the information which of the two inputs $x$ and $y$ are the $0$- and $1$-inputs of $f$. Indeed, we assume that the bits queried in the first entry of the tuple, always correspond to the $0$-input that defined the sabotaged input $z$. We remark here that we can always remove this assumption if we allow for additive overhead of $O(\Q(f))$, after all we can always run a quantum algorithm that computes $f$ on the first bits from all our queried tuple, to compute $f(x)$, and thus finding out whether it was a $0$- or $1$-input to begin with.

    Having defined the sabotage functions/relations and the input models we now define four different notions of quantum sabotage complexity of $f$.

    \begin{defi}[Sabotage complexity]\label{defi:sabotage}
        Let $n$ be a positive integer, $D \subseteq \zone^n$, and let $f : D \to \zone$ be a (partial) Boolean function. Define,
        \begin{align*}
            \QS_{\weak}(f)\coloneqq \Q(f_{\sab}), & \qquad \QS_\weak^\ind(f) \coloneqq \Q(f_\sab^\ind)\\
            \QS_{\strong}(f)\coloneqq \Q(f_{\sab}^\strong), & \qquad \QS_\strong^\ind(f) \coloneqq \Q(f_\sab^{\strong, \ind}).
        \end{align*}
        Analogously, define
        \begin{align*}
            \RS_{\weak}(f)\coloneqq \R(f_{\sab}), & \qquad \RS_\weak^\ind(f) \coloneqq \R(f_\sab^\ind)\\
            \RS_{\strong}(f)\coloneqq \R(f_{\sab}^\strong), & \qquad \RS_\strong^\ind(f) \coloneqq \R(f_\sab^{\strong, \ind}).
        \end{align*}
    \end{defi}

    \subsection{Randomized sabotage complexity}

    By comparing the definitions in Definition~\ref{defi:sabotage} to those in \cite{BK18}, we observe that $\RS = \RS_{\weak}$. However, we also argue that for all functions, the other randomized complexity measures are equal up to constants.

    \begin{proposition}\label{prop: all sab equal randomized}
        Let $n$ be a positive integer, and $D \subseteq \zone^n$. Let $f : D \to \zone$ be a (partial) Boolean function. Then,
        \[\RS(f) \coloneqq \RS_{\weak}(f) = \Theta(\RS_{\strong}(f)) = \Theta(\RS_{\weak}^{\ind}(f)) = \Theta(\RS_{\strong}^{\ind}(f)).\]
    \end{proposition}

    \begin{proof}
        It is clear that we have $\RS_{\strong}(f) \leq \RS_{\weak}(f)$, and $\RS_{\strong}^{\ind}(f) \leq \RS_{\weak}^{\ind}(f)$ because the input model is strictly stronger. For the opposite directions, suppose that we have an algorithm $\mathcal{A}$ in the strong input model. Let $x \in f^{-1}(0)$ and $y \in f^{-1}(1)$. Note that for every $j \in [n]$ on queries where $x_j = y_j$, the oracle outputs $(0,0,0)$ or $(1,1,1)$. Thus, on these indices, the same algorithm in the weak model would yield the outputs $0$ or $1$, respectively. Moreover, with high probability $\mathcal{A}$ must query a $j \in [n]$ where $x_j \neq y_j$ at some point, since otherwise it cannot distinguish $\sabStar{x}{y}$ and $\sabDagger{x}{y}$. This gives us the following $\RS_\weak^\ind$ algorithm: run $\mathcal{A}$ constantly many times (with strong queries replaced by weak queries) until we query a $j \in [n]$ where $x_j \neq y_j$. At that point, we can interrupt the algorithm and we have enough information to solve the sabotage problem, in both the decision and index model. Thus $\RS_{\weak}(f) = O(\RS_{\strong}(f))$, and $\RS_{\weak}^{\ind}(f) = O(\RS_{\strong}^{\ind}(f))$.

        It remains to show that $\RS_{\strong}(f) = \Theta(\RS_{\strong}^{\ind}(f))$. By the definitions of the models, we have $\RS_{\strong}(f) \leq \RS_{\strong}^{\ind}(f)$. On the other hand, suppose that we have an algorithm $\mathcal{B}$ that figures out whether we have a $*$- or $\dagger$-input. Then, by the same logic as given in the previous paragraph, with high probability, $\mathcal{B}$ must have encountered at least one $*$ or $\dagger$, so we can read back in the transcript to find out at which position it made that query. Repeating $\mathcal{B}$ a constant number of times this way yields $\RS_{\strong}^{\ind}(f) = O(\RS_{\strong}(f))$.
    \end{proof}
    Note that we did not do a formal analysis of success probabilities in the above argument, but this is not hard to do. We omit precise details for the sake of brevity.

    \subsection{Quantum sabotage complexity}

    In the quantum case, we have not been able to prove equivalences between the different definitions. However, we can still prove some bounds between them. We refer the reader to Figure~\ref{fig:relations} for a pictorial representation of all relationships.

    \begin{proposition}
        \label{prop:relations-quantum-measures}
        Let $n$ be a positive integer, $D \subseteq \zone^n$, and let $f : D \to \zone$ be a (partial) Boolean function. Then,
        \[\QS_{\strong}(f) = O(\QS_{\strong}^{\ind}(f)) = O(\QS_{\weak}^{\ind}(f)) = O(\min\{\RS(f)),\sqrt{n}\}),\]
        and
        \[\QS_{\strong}(f) = O(\QS_{\weak}(f)) = O(\QS_{\weak}^{\ind}(f)).\]
    \end{proposition}

    \begin{proof}
        Just as in the randomized case, $\QS_{\strong}(f) = O(\QS_{\weak}(f))$ and $\QS_{\strong}^{\ind}(f) = O(\QS_{\weak}^{\ind}(f))$, because the input model is stronger, i.e., we can simulate a query to the weak oracle with $O(1)$ queries to the strong oracle. Furthermore, $\QS_{\strong}(f) = O(\QS_{\strong}^{\ind}(f))$ and $\QS_{\weak}(f) = O(\QS_{\weak}^{\ind}(f))$, because once we have found a $j \in [n]$ where $x_j \neq y_j$, we can query that index with one more query to find figure out whether we have a $*$-input or a $\dagger$-input.

        In the weak input model, we can always simply run Grover's algorithm to find a position $j \in [n]$ where $z_j \in \{*,\dagger\}$. This takes $O(\sqrt{n})$ queries. Thus, it remains to show that $\QS_{\weak}^{\ind}(f) = O(\RS(f))$. We showed in the previous proposition that $\RS(f)$ is the same up to constants to $\RS_{\weak}^{\ind}(f)$, and since the quantum computational model is only stronger than the randomized one, we find $\QS_{\weak}^{\ind}(f) = O(\RS_{\weak}^{\ind}(f)) = O(\RS(f))$.
    \end{proof}

    \section{Upper bounds on $\QS$}

    In this section, we prove upper bounds on the complexity measures introduced in Section~\ref{sec:QuantumSabotageComplexity}. We start by showing that the quantum sabotage complexity in the strong model is upper bounded by the regular bounded-error quantum query complexity. Thereby, we prove that $\QS_{\strong}(f)$ has the property that was sought for in \cite{BK18}, i.e., in this model computing $f_{\sab}$ indeed costs at most as many queries as computing $f$ itself.

    \qsvsq*

    \begin{proof}
        Let $\mathcal{A}$ be a bounded-error quantum query algorithm that computes $f$. We construct a quantum query algorithm $\mathcal{B}$ that computes $f_{\sab}$ in the strong input model.

        Recall that in the strong input model (as in Definition~\ref{def:StrongInputModel}), our input is viewed as $w = ((x_j, y_j, z_j))_{j=1}^n$ where $f(x) = 0, f(y) = 1$ and $z$ is the sabotaged input constructed from $x$ and $y$. A query on the $j$th position to the oracle $O_w^{\strong}$ returns a tuple $(x_j, y_j, z_j)$. Now, we define $\mathcal{B}$ to be the same algorithm as $\mathcal{A}$, but whenever $\mathcal{A}$ makes a query, it performs the following operation instead:

        \begin{enumerate}
            \setlength\itemsep{-.3em}
            \item Query $O_w^{\strong}$, denote the outcome by $(x_j,y_j,z_j)$.
            \item If $z_j \in \{0,1\}$, return $z_j$.
            \item If $z_j = *$, return $x_j$.
            \item If $z_j = \dagger$, return $y_j$.
        \end{enumerate}

        Note that this operation can indeed be implemented quantumly making $2$ queries to $O_w^{\strong}$. The initial query performs the instructions described above, and the second query uncomputes the values from the tuple we don't need for the rest of the computation. Note here that $(O_w^\strong)^4 = I$, and thus $(O_w^\strong)^3 = (O_w^\strong)^{-1}$, which is what we need to implement for the uncomputation operations.

        We observe that the above operation always returns $x_j$ whenever we have a $*$-input, and $y_j$ whenever we have a $\dagger$-input. Thus, if we run algorithm $\mathcal{B}$ with this oracle operation (in superposition, including uncomputation), then we output $f(x) = 0$ on a $*$-input, and $f(y) = 1$ on a $\dagger$-input, with the same success probability as that of $\mathcal{A}$. As this query operation can be implemented with a constant number of calls to the query oracle $O_w^\strong$, we conclude that $\QS_{\strong}(f) = O(\Q(f))$.
    \end{proof}

    It is not obvious how we can modify the above algorithm to also output the index where $x$ and $y$ differ. However, we can design such an algorithm using very different techniques, and give an upper bound on $\QS_{\strong}^{\ind}(f)$ in terms of $\QD(f)$. To that end, we first prove a fundamental lemma that is similar to~\cite[Lemma~12]{BK19}.

    \hybrid*

    \begin{proof}
        We write $y = x_B$, and we let $\ket{\psi_x^t}$ and $\ket{\psi_y^t}$ be the states right before the $t$h query, when we run $\mathcal{A}$ on inputs $x$ and $y$, respectively. We also let $\ket{\psi_x}$ and $\ket{\psi_y}$ be the final states of the algorithm run on $x$ and $y$, respectively. Since $\mathcal{A}$ can distinguish $x$ and $y$ with high probability, we observe that $\ket{\psi_x}$ and $\ket{\psi_y}$ must be far apart, i.e., their inner product must satisfy
        \[1 - |\braket{\psi_x}{\psi_y}| = \Omega(1).\]

        For every $j \in [n]$, let $\mathcal{H}_j$ be the subspace of the state space of $\mathcal{A}$ that queries the $j$th bit of the input. In other words, we let $\mathcal{H}_j$ be the span of all states that pick up a phase of $(-1)^{x_j}$, when the algorithm $\mathcal{A}$ calls the oracle $O_x$. We let $\Pi_j$ be the projector on this subspace.

        Next, we let $\mathcal{H}_B$ be the subspace that contains all $\mathcal{H}_j$'s with $j \in B$. In other words, we write $\mathcal{H}_B = \oplus_{j \in B} \mathcal{H}_j$. We immediately observe that the projector onto $\mathcal{H}_B$, denoted by $\Pi_B$, satisfies $\Pi_B = \sum_{j \in B} \Pi_j$. Note that it is exactly the subspace $\mathcal{H}_B$ on which the oracles $O_x$ and $O_y$ act differently. So, intuitively, if a state $\ket{\psi}$ has a big component in $\mathcal{H}_B$, then $O_x\ket{\psi}$ and $O_y\ket{\psi}$ will be far apart from each other.

        Now, we define $p_{x,t} := \norm{\Pi_B\ket{\psi_x^t}}^2$, i.e., the squared overlap of the state $\ket{\psi_x^t}$ with the subspace $\mathcal{H}_B$. Intuitively, if we were to interrupt the algorithm $\mathcal{A}$ run on input $x$ right before the $t$th query, and we were to measure the query register, then the probability of measuring a $j \in B$ is $p_{x,t}$.

        The crucial observation that we make is that
        \begin{align*}
            &\left|\braket{\psi_x^t}{\psi_y^t}\right| - \left|\braket{\psi_x^{t+1}}{\psi_y^{t+1}}\right| \leq \left|\braket{\psi_x^t}{\psi_y^t} - \braket{\psi_x^{t+1}}{\psi_y^{t+1}}\right| = \left|\braket{\psi_x^t}{\psi_y^t} - \bra{\psi_x^t}O_x^{\dagger}O_y\ket{\psi_y^t}\right| \\
            &\qquad = \left|\bra{\psi_x^t}\left(I - O_x^{\dagger}O_y\right)\ket{\psi_y^t}\right| = 2\left|\bra{\psi_x^t}\Pi_B\ket{\psi_y^t}\right| \leq 2\norm{\Pi_B\ket{\psi_x^t}} \cdot \norm{\Pi_B\ket{\psi_y^t}} \\
            &\qquad = 2\sqrt{p_{x,t} \cdot p_{y,t}} \leq p_{x,t} + p_{y,t}.
        \end{align*}
        Here, we used the triangle inequality, the Cauchy-Schwarz inequality, and the AM-GM inequality, in order. Finally, we observe that the initial states for algorithm $\mathcal{A}$ run on $x$ and $y$ are the same, and so $|\braket{\psi_x^1}{\psi_y^1}| = 1$. Thus, identifying $\ket{\psi_x^{T+1}}$ with $\ket{\psi_x}$, and similarly for $y$, we obtain that
        \[1 - |\braket{\psi_x}{\psi_y}| = \sum_{t=1}^T |\braket{\psi_x^t}{\psi_y^t}| - |\braket{\psi_x^{t+1}}{\psi_y^{t+1}}| \leq \sum_{t=1}^T (p_{x,t} + p_{y,t}).\qedhere\]
    \end{proof}

    We now show how the above lemma can be used to prove a connection between $\QS_{\strong}^{\ind}(f)$ and $\QD(f)$.

    \begin{theorem}
        \label{thm:qsvsqd}
        Let $f : \{0,1\}^n \to \{0,1\}$. We have $\QS_{\strong}^{\ind}(f) = O(\QD(f)^{3/2})$.
    \end{theorem}

    \begin{proof}
        Suppose we have an algorithm $\mathcal{A}$ that distinguishes between input $x \in f^{-1}(0)$ and $y \in f^{-1}(1)$ with high probability in $T$ queries. Then, we construct an algorithm $\mathcal{B}$ that works in the strong sabotage input model, and finds an index $j \in [n]$ where $x_j \neq y_j$.

        First, consider the following procedure. We pick an input $x$ or $y$ with probability $1/2$, and we pick a time step $t \in \{1, \dots, T\}$ uniformly at random. We run $\mathcal{A}$ until right before the $t$th query, and then we measure the query register to obtain an index $j \in [n]$. From Lemma~\ref{lem: hybrid new}, we obtain that the probability that $x_j \neq y_j$ is lower bounded by $\Omega(1/T)$.\footnote{It is important to remark here that we have access to our strong sabotage oracle now (Definition~\ref{def:StrongInputModel}). Recall that there are four registers: $\ket{j}, \ket{b_x}, \ket{b_y}, \ket{b_z}$. When we say ``run $\mathcal{A}$'' with this oracle, we mean all operations act as identity on the last register.} Thus, running this algorithm $O(T)$ times would suffice to find a $j \in [n]$ such that $x_j \neq y_j$, with high probability.

        However, we can do slightly better than that. Note that once the algorithm gives us an index $j \in [n]$, it takes just one query (to $z$) to find out if $z_j \in \cbra{*, \dagger}$. Thus, we can use amplitude amplification, and use $O(\sqrt{T})$ iterations of the above procedure, to find a $j \in [n]$ such that $z_j \in \cbra{*, \dagger}$ (equivalently, $x_j \neq y_j$). Each application of the procedure takes $O(T)$ queries to implement in the worst case. Thus, the final query complexity is $O(T^{3/2})$.
    \end{proof}

    Combining the above result with \cite[Proposition~6]{BK19}, which states that $\QD(f) = O(\Q(f))$, immediately yields the following theorem.

    \qsindvsq*

    \section{$\QS_\strong(f)$ vs.~$\sqrt{\fbs(f)}$}

    So far we have shown upper bounds on quantum sabotage complexity. In this section we show our lower bounds. We first show that $\QS_\strong(f)$ (and hence $\QS_\strong^\ind(f), \QS_\weak(f), \QS_\weak^\ind(f)$ as well) is bounded from below by $\sqrt{\fbs(f)}$. This is a generalization of the known bound of $\Q(f) = \Omega(\sqrt{\bs(f)})$~\cite{BBCMW01}. In particular, this already implies that quantum sabotage complexity is polynomially related to quantum query complexity for all total Boolean functions $f$. Next observe that, unlike in the usual quantum query setting, this $\sqrt{\fbs(f)}$ lower bound is \emph{tight} for standard functions such as Or, And, Majority and Parity because of the Grover-based $O(\sqrt{n})$ upper bound on the quantum sabotage complexity of all functions. This suggests the possibility of the quantum sabotage complexity of $f$ being $\Theta(\sqrt{\fbs(f)})$ for all total $f$. In the next subsection we rule this out, witnessed by $f$ as the Indexing function, for which we show the quantum sabotage complexity to be $\Theta(\fbs(f))$.

    \subsection{A general lower bound}

    In the appendix we show that $\RS(f) = \Omega(\fbs(f))$, as well as the quantum bound of $\Q(f) = \Omega(\sqrt{\fbs(f)})$.
    We now wish to follow the same approach as in the proof of the latter bound for proving a lower bound on $\QS_{\strong}(f)$. To that end, we know that $\QS_{\strong}(f) \coloneqq \Q(f_\sab^\strong)= \Omega(\ADV^+(f_{\sab}^{\strong}))$~\cite{Amb06}, so it remains to show that $\ADV^+(f_{\sab}^{\strong}) = \Omega(\sqrt{\fbs(f)})$. Thus, we adapt the proof of Lemma~\ref{lem:fbsvsadv} in the sabotaged setting.

    \begin{lemma}
        Let $n$ be a positive integer, $D \subseteq \zone^n$ and $f : D \to \zone$. Then, $\ADV^+(f_{\sab}^{\strong}) = \Omega(\sqrt{\fbs(f)})$.
    \end{lemma}

    \begin{proof}
        Let $x$ be the instance for which $\fbs(f) = \fbs(f,x)$, and let $(w_y)_{y \in Y}$ be the optimal weight assignment (see Definition~\ref{def:fbs}). Similar to the proof of Lemma~\ref{lem:fbsvsadv}, we generate a (non-negative weight) adversary matrix, $\Gamma \in \mathbb{R}^{D_{\sab}^{\strong} \times D_{\sab}^{\strong}}$. We define $\Gamma$ to be the all-zeros matrix, except for the instances where $((x,y,*),(x,y',\dagger))$ and $((x,y,\dagger),(x,y',*))$, with $y,y' \in Y$, where we define it to be
        \[\Gamma[(x,y,*),(x,y',\dagger)] = \Gamma[(x,y',\dagger),(x,y,*)] = \sqrt{w_yw_{y'}}.\]

        $\Gamma$ has a simple sparsity pattern. Only the rows and columns that are labeled by $(x,y,*)$ and $(x,y,\dagger)$, with $y \in Y$, are non-zero. Additionally, observe from the definition of the matrix entries that for any $y,y' \in Y$, we have
        \[\Gamma[(x,y,*),(x,y',\dagger)] = \sqrt{w_yw_{y'}} = \Gamma[(x,y,\dagger),(x,y',*)].\]
        Thus, in every $2 \times 2$-block formed by rows $(x,y,*)$ and $(x,y,\dagger)$ and columns $(x,y',*)$ and $(x,y',\dagger)$, we have that the two diagonal elements are $0$, and the two off-diagonal elements are equal. Hence, by removing unimportant rows and columns that are completely zero, we can rewrite $\Gamma$ as
        \[\Gamma = \Gamma' \otimes \begin{bmatrix}
            0 & 1 \\
            1 & 0
        \end{bmatrix}, \qquad \text{where} \qquad \Gamma' \in \mathbb{R}^{Y \times Y}, \qquad \text{with} \qquad \Gamma'[y,y'] = \sqrt{w_yw_{y'}}.\]

        It now becomes apparent that $\Gamma'$ is of rank $1$. Indeed, it is the outer product of a vector $\sqrt{w} \in \mathbb{R}^Y$, that contains the entries $\sqrt{w_y}$ at every index labeled by $y$. From some matrix arithmetic, we now obtain
        \[\norm{\Gamma} = \norm{\Gamma'} \cdot 1 = \norm{\sqrt{w}\sqrt{w}^T} = \norm{\sqrt{w}}^2 = \sum_{y \in Y} w_y = \fbs(f).\]

        Thus, it remains to prove that $\norm{\Gamma \circ \Delta_j} = O(\sqrt{\fbs(f)})$, for all $j \in [n]$. Indeed, then we can scale our matrix $\Gamma$ down by $\Theta(\sqrt{\fbs(f)})$ so that it is feasible for the optimization program in Definition~\ref{def:adv}, and the objective value will then become $\Theta(\sqrt{\fbs(f)})$ as predicted.

        Let $j \in [n]$. To compute $\norm{\Gamma \circ \Delta_j}$, we look at its sparsity pattern. Observe that whenever we query the $j$th bit of $(x,y,*)$ and $(x,y',\dagger)$, we obtain the tuples $(x_j,y_j,z_j)$ and $(x_j,y_j',z_j')$. If $y_j \neq y_j'$, then it is clear that these tuples are not the same. Similarly, if $x_j \neq y_j = y_j'$, then both $z_j$'s will be different, as one will be a $*$ and the other will be a $\dagger$. Thus, the queried tuples are only identical whenever $x_j = y_j = y_j'$, which implies that we can rewrite
        \[\Gamma \circ \Delta_j = \Gamma'_j \otimes \begin{bmatrix}
            0 & 1 \\
            1 & 0
        \end{bmatrix}, \quad \text{where} \quad \Gamma_j' \in \mathbb{R}^{Y \times Y}, \quad \text{with} \quad \Gamma_j'[y,y'] = \begin{cases}
            0, & \text{if } y_j = y_j' = x_j, \\
            \sqrt{w_yw_{y'}}, & \text{otherwise}.
        \end{cases}\]
        We now let $Y_0 = \{y \in Y : x_j \neq y_j\}$ and $Y_1 = \{y \in Y : x_j = y_j\}$. We interpret $\Gamma'_j$ as a $2 \times 2$-block matrix, where the first rows and columns are indexed by $Y_0$ and the last ones are indexed by $Y_1$. Then, $\Gamma_j'$ takes on the shape
        \[\Gamma_j' = \begin{bmatrix}
            A & B \\
            B^T & 0
        \end{bmatrix},\]
        with $A \in \mathbb{R}^{Y_0 \times Y_0}$ and $B \in \mathbb{R}^{Y_0 \times Y_1}$, defined as
        \[A[y,y'] = \sqrt{w_yw_{y'}}, \qquad \text{and} \qquad B[y,y'] = \sqrt{w_yw_{y'}}.\]
        We observe that
        \[\norm{\Gamma \circ \Delta_j} = \norm{\Gamma_j'} \cdot 1 \leq \norm{A} + \norm{B},\]
        and hence it remains to compute $\norm{A}$ and $\norm{B}$.

        We now define the vectors $\sqrt{w_0} \in \mathbb{R}^{Y_0}$ and $\sqrt{w_1} \in \mathbb{R}^{Y_1}$, defined by $(\sqrt{w_0})_y = \sqrt{w_y}$, and $(\sqrt{w_1})_y = \sqrt{w_y}$. We observe that $A = \sqrt{w_0}\sqrt{w_0}^T$, and $B = \sqrt{w_0}\sqrt{w_1}^T$. Thus, we obtain
        \[\norm{A}^2 = \norm{\sqrt{w_0}\sqrt{w_0}^T}^2 = \norm{\sqrt{w_0}}^4 = \left[\sum_{y \in Y_0} w_y\right]^2 = \left[\sum_{\substack{y \in Y \\ x_j \neq w_j}} w_y\right]^2 \leq 1,\]
        and similarly
        \[\norm{B}^2 = \norm{\sqrt{w_0}\sqrt{w_1}^T}^2 = \norm{\sqrt{w_0}}^2 \cdot \norm{\sqrt{w_1}}^2 = \left[\sum_{\substack{y \in Y \\ x_j \neq y_j}} w_y\right] \cdot \left[\sum_{\substack{y \in Y \\ x_j = y_j}} w_y\right] \leq 1 \cdot \fbs(f).\qedhere\]
    \end{proof}

    The proof of Theorem~\ref{thm: qsab bs lower bound} now follows as a simple corollary from this lemma and the fact that $\QS_{\strong}(f) \coloneqq \Q(f_\sab^\strong) = \Omega(\ADV^+(f_{\sab}^{\strong}))$, where the last bound follows from the positive-weighted adversary lower bound of Ambainis~\cite{Amb06}.

    \subsection{A stronger lower bound for Indexing}

    As observed in the discussion following Theorem~\ref{thm: main 1}, we have $\QS_\weak^\ind(f)$ (and hence $\QS_\weak(f)$ and $\QS_\strong^\ind(f)$ and $\QS_\strong(f)$) is $O(\sqrt{n})$ for all $f : \zone^n \to \zone$. In particular, the $\sqrt{\fbs(f)}$ lower bound for $\QS_\strong(f)$ is tight for standard functions like And, Or, Parity, Majority. In view of this it is feasible that $\QS_\strong(f) = O(\sqrt{\fbs(f)})$ for all total $f : \zone^n \to \zone$. In the remaining part of this section, we rule this out, witnessed by the Indexing function.

    We use Ambainis's adversary method to prove lower bounds on quantum query complexity~\cite{Amb02}.

    \begin{lemma}[{\cite[Theorem~6.1]{Amb02}}]\label{lem: ambainis v2}
        Let $f : \cbra{0, 1, *, \dagger}^k \to \zone$ be a (partial) Boolean function. Let $X, Y \subseteq \cbra{0, 1, *, \dagger}^k$ be two sets of inputs such that $f(x) \neq f(y)$ if $x \in X$ and $y \in Y$. Let $R \subseteq X \times Y$ be nonempty, and satisfy:
        \begin{itemize}
            \item For every $x \in X$ there exist at least $m_X$ different $y \in Y$ such that $(x, y) \in R$.
            \item For every $y \in Y$ there exist at least $m_Y$ different $x \in X$ such that $(x, y) \in R$.
        \end{itemize}
        Let $\ell_{x, i}$ be the number of $y \in Y$ such that $(x, y) \in R$ and $x_i \neq y_i$, and similarly for $\ell_{y, i}$. Let $\ell_{\max} = \max_{i \in [k]}\max_{(x, y) \in R, x_i \neq y_i}\ell_{x,i}\ell_{y,i}$. Then any algorithm that computes $f$ with success probability at least $2/3$ uses $\Omega\rbra{\sqrt\frac{m_X m_Y}{\ell_{\max}}}$ quantum queries to the input function.
    \end{lemma}

    Define the Indexing function as follows.
    For a positive integer $n > 0$, define the function $\IND_n : \zone^n \times \zone^{2^n} \to \zone$ as $\IND_n(a, b) = b_{\bin(a)}$, where $\bin(a)$ denotes the integer in $[2^n]$ whose binary representation is $a$.

    We first note that the fractional block sensitivity is easily seen to be bounded from below by block sensitivity, and bounded from above by deterministic (in fact randomized) query complexity. Both the block sensitivity (in fact, sensitivity) and deterministic query complexity of $\IND_n$ are easily seen to be $n+1$, implying $\fbs(\IND_n) = n+1$.

    \begin{theorem}\label{thm:sabIND}
        Let $n$ be a positive integer. Then, $\QS_{\weak}(\IND_n) = \Omega(n)$.
    \end{theorem}

    \begin{proof}
        Recall from Definition~\ref{defi:sabotage} that $\QS_\weak(\IND_n) = \Q(\IND_{n,\sab})$.
        We construct a hard relation for $f = \IND_{n,\sab}$ and use Lemma~\ref{lem: ambainis v2}. Recall that this relation must contain pairs of inputs. For each pair, the function must evaluate to different outputs. For ease of notation we first define the pairs of inputs $(a_1, b_1), (a_2, b_2)$ in the relation, and then justify that these inputs are indeed in $S_*$ and $S_\dagger$, respectively.

        Define $(a_1, b_1), (a_2, b_2) \in R$ if and only if all of the following hold true:
        \begin{enumerate}
            \item $(a_1, b_1) \in f^{-1}(0)$ (i.e., $(a_1, b_1) \in S_*$ for $\IND_n$), $(x_2, y_2) \in f^{-1}(1)$ (i.e., $(a_2, b_2) \in S_\dagger$ for $\IND_n$),
            \item $|a_1 \oplus a_2| = 2$, i.e., the Hamming distance between $a_1$ and $a_2$ is 2,
            \item $b_1$ is all-0, except for the $\bin(a_1)$'th index, which is $*$,
            \item $b_2$ is all-0, except for the $\bin(a_2)$'th index, which is $\dagger$.
        \end{enumerate}
        First note that $(a_1, b_1) = \sabStar{(a_1, 0^{2^n})}{(a_1, e_{a_1})}$, where $e_{a_1} \in \zone^{2^n}$ is the all-0 string except for the $\bin(a_1)$'th location, which is a 1. Similarly, $(a_2, b_2) = \sabDagger{(a_2, 0^{2^n})}{(a_2, e_{a_2})}$. Thus, $(a_1, b_1)$ and $(a_2, b_2)$ are in $S_*$ and $S_\dagger$, respectively.
        In particular, in the language of Lemma~\ref{lem: ambainis v2} we have
        \begin{align}
            \label{eq:XforR}
            X = \cbra{(a, b) \in \zone^n \times \{0,1,*\}^{2^n} : b~\textnormal{is all-0 except for the~}\bin(a)\textnormal{'th index, which is}~*}.
        \end{align}
        Similarly,
        \begin{align}
            \label{eq:YforR}
            Y = \cbra{(a, b) \in \zone^n \times \{0,1,\dagger\}^{2^n} : b~\textnormal{is all-0 except for the~}\bin(x)\textnormal{'th index, which is}~\dagger}.
        \end{align}

        We now analyze the quantities $m_X, m_Y$ and $\ell_{\max}$ from Lemma~\ref{lem: ambainis v2}.
        \begin{enumerate}
            \item $m_X = m_Y = \binom{n}{2}$: Consider $(a, b) \in X$. The number of elements $(a', b') \in Y$ such that $((a, b), (a', b') \in R)$ is simply the number of strings $a'$ that have a Hamming distance of 2 from $a$, since each such $a'$ corresponds to exactly one $b'$ with $((a, b), (a', b') \in R)$, where $b'$ is the all-0 string except for the $\bin(a)'$th location which is a $\dagger$. Thus $m_X = \binom{n}{2}$. The argument for $m_Y$ is essentially the same.
            \item $\ell_{\max} = \min\cbra{\binom{n}{2}, (n-1)^2}$: We consider two cases.
            \begin{enumerate}
                \item\label{item:proof} $i \in [n]$: Fix $((a, b), (a', b')) \in R$ with $a_i \neq a'_i$. Recall that $\ell_{(a, b), i}$ is the number of $(a'', b'') \in Y$ such that $((a, b), (a'', b'')) \in R$ and $a_i \neq a''_i$. Following a similar logic as in the previous argument, this is simply the number of $a''$ that have Hamming distance 2 from $a$ and additionally satisfy $a_i \neq a''_i$. There are $n-1$ possible locations for the other difference between $a$ and $a''$, so $\ell_{(a, b), i} = n-1$ in this case. Essentially the same argument shows $\ell_{(a', b'), i} = n-1$, and so $\ell_{(a, b), i} \cdot \ell_{(a', b'), i} = (n-1)^2$.
                \item $i \in [2^n]$: Fix $((a, b), (a', b')) \in R$ with $b_i \neq b'_i$. By the structure of $R$, $b$ is the all-0 string except for the $\bin(a)$'th location which is a $*$, and $b'$ is the all-0 string except for the $\bin(a')$'th location which is a $\dagger$. Thus $i \in \cbra{\bin(a), \bin(a')}$. Without loss of generality, assume $i = \bin(a)$, and thus $b_{\bin(a)} = *$. For each $a''$ with Hamming distance 2 from $a$, we have $((a, b), (a'', b'')) \in R$ where $b''$ is all-0 except for the $a''$th index which is $\dagger.$ In particular, $b''_{\bin(a)} = 0$. So we have $\ell_{(a, b), i} = \binom{n}{2}$. On the other hand, we have $b'_{\bin(a)} = 0$. So the only $(a'', b'')$ with $((a'', b''), (a', b')) \in R$ and with $b''_{\bin(a)} \neq b'_{\bin(a)}$ is $(a'', b'') = (a, b)$. Thus $\ell_{(a', b'), i} = 1$.
            \end{enumerate}
        \end{enumerate}
        Lemma~\ref{lem: ambainis v2} then implies
        \begin{align}
            \QS_\weak(\IND_n) = \Q(\IND_{n, \sab}) = \Q(f) = \Omega\rbra{\sqrt\frac{\binom{n}{2}^2}{\min\cbra{\binom{n}{2}, (n-1)^2}}} = \Omega(n),
        \end{align}
        proving the theorem.
    \end{proof}

    This proof can easily be adapted to also yield a lower bound of $\QS_\strong(\IND_n) = \Omega(n)$. We include a proof in the appendix for completeness.
    As a corollary we obtain the following.

    \begin{corollary}
        \label{thm:sabINDothermodels}
        Let $n$ be a positive integer. Then, $\QS_{\strong}^{\ind}(\IND_n)=\Omega(n)$ and $\QS_{\weak}^{\ind}(\IND_n)=\Omega(n)$.
    \end{corollary}

    \section{Open questions}
    In this paper we studied the quantum sabotage complexity of Boolean functions, which we believe is a natural extension of randomized sabotage complexity introduced by Ben-David and Kothari~\cite{BK18}. We note here that in a subsequent work~\cite{BK19} they also defined a quantum analog, but this is fairly different from the notion studied in this paper.

    We argued, by showing several results, that it makes sense to consider four different models depending on the access to input, and output requirements. While it is easily seen that the randomized sabotage complexity of a function remains (asymptotically) the same regardless of the choice of model (see Proposition~\ref{prop: all sab equal randomized}), such a statement is not clear in the quantum setting. In our view, the most interesting problem left open from our work is to prove or disprove that even the quantum sabotage complexity of a function is asymptotically the same in all of these four models. It would also be interesting to see tight polynomial relationships between the various quantum sabotage complexities and quantum query complexity.

    \bibliography{main}

    \appendix

    \section{Lower bounds on $\QS(f)$ and $\RS(f)$ in terms of $\fbs(f)$}

    To the best of our knowledge, the only lower bound on randomized sabotage complexity $\RS(f)$ in the existing literature is $\Omega(\QD(f))$~\cite[Corollary~11]{BK19}. We remark that $\RS(f)$ can also be lower bounded by $\fbs(f)$, as is witnessed by the lemma below. We do not claim this to be a new result, but we have not found a formal reference to such a statement. We include a proof for completeness.

    \begin{lemma}
        \label{lem:fbsrs}
        $\fbs(f) = O(\RS(f))$.
    \end{lemma}

    \begin{proof}
        From \cite[Theorem~3.3]{BK18}, we find that $\R(f_{\sab}) = \Theta(\RS(f))$, and so it suffices to show that $\fbs(f) \leq 10\R(f_{\sab})$. By Yao's minimax principle, it suffices to exhibit a distribution over inputs for which any deterministic algorithm fails to compute $f_{\sab}$ correctly with probability at least $2/3$ in less than $\fbs(f)/10$ queries.

        Without loss of generality, let $x \in f^{-1}(0)$ be the instance for which fractional block-sensitivity is maximized, let $Y = f^{-1}(1)$, and let $(w_y)_{y \in Y}$ be the optimal weight assignment in the fractional block sensitivity linear program (see Definition~\ref{def:fbs}). We can similarly think of every $y \in Y$ as obtained by flipping the bits in $x$ that belong to a sensitive block $B$. We denote $y = x_B$, and write $w_y = w_B$.

        Let $T$ be a deterministic algorithm for $f_\sab$ with query complexity less than $\fbs(f)/10$, and let $\mu$ be a distribution defined as follows:
        for all $B \subseteq [n]$ that is a sensitive block for $f$, let $\mu([x,x_B,*]) = \mu([x,x_B,\dagger]) = w_B/(2\fbs(f))$. Since $\sum_{B \subseteq [n]}w_B = \fbs(f)$ where the sum is over all sensitive blocks for $f$, $\mu$ forms a legitimate probability distribution.

        Consider a leaf $L$ of $T$, and suppose that the output at $L$ is $b \in \zone$. By the definition of fractional block sensitivity, we know that for any $j \in [n]$, we have $\sum_{B : j \in B} w_B \leq 1$. Thus, by a union bound, the $\mu$-mass of inputs supported by our distribution (of the form $[x,x_B,*]$ or $[x,x_B,\dagger]$) reaching $L$, and such that no bit of $B$ has been read on this path is at least $1 - \rbra{\frac{\fbs(f)}{10} \cdot \frac{1}{\fbs(f)}} \geq 9/10$. Note that for each such $B$, both the input $[x, x_B, *]$ and $[x, x_B, \dagger]$ reach $L$. Since $\mu$ allotted equal mass to each such pair, this means an error is made at $L$ for inputs with $\mu$-mass at least $9/20 > 1/3$, concluding the proof.
    \end{proof}

    We now turn to lower bounding $\QS_{\strong}(f)$ by $\Omega(\sqrt{\fbs(f)})$. To that end, recall that for any Boolean function $f$, we have the chain of inequalities $\Q(f) = \Omega(\ADV^+(f)) = \Omega(\sqrt{\fbs(f)})$. The first inequality can be found in~\cite{SS06}, for example. While we believe the second inequality is known, we were again unable to find a published proof of it. We include a proof here for completeness.

    \begin{lemma}
        \label{lem:fbsvsadv}
        Let $n$ be a positive integer, $D \subseteq \zone^n$, and $f : D \to \zone$ be a (partial) Boolean function. Then, $\sqrt{\fbs(f)} = O(\ADV^+(f))$.
    \end{lemma}

    \begin{proof}
        Recall the definition of fractional block sensitivity, Definition~\ref{def:fbs}. Let $x$ be the input to $f$ for which $\fbs(f,x) = \fbs(f)$. Let $(w_y)_{y \in Y}$ be the optimal solution of the linear program. We now define an adversary matrix $\Gamma$, that forms a feasible solution to the optimization program in Definition~\ref{def:adv}. In particular, we define $\Gamma \in \mathbb{R}^{D \times D}$ as the all-zeros matrix, except for the entries $(x,y)$ and $(y,x)$ where $y \in Y$, which we define to be
        \[\Gamma[x,y] = \Gamma[y,x] = \sqrt{w_y}.\]

        The matrix $\Gamma$ we constructed is very sparse. It is only non-zero in the row and column that is indexed by $x$. Moreover, $\Gamma[x,x] = 0$ as well. In particular, this makes computing its norm quite easy, we just have to compute the $\ell_2$-norm of the column indexed by $x$.

        For any $j \in [n]$, we thus find
        \[\norm{\Gamma \circ \Delta_j}^2 = \sum_{\substack{y \in Y \\ x_j \neq y_j}} \Gamma[x,y]^2 = \sum_{\substack{y \in Y \\ x_j \neq y_j}} w_y \leq 1,\]
        and so $\Gamma$ is a feasible solution to the optimization program in Definition~\ref{def:adv}. Finally, we have
        \[\ADV^+(f)^2 \geq \norm{\Gamma}^2 = \sum_{y \in Y} \Gamma[x,y]^2 = \sum_{y \in Y} w_y = \fbs(f).\qedhere\]
    \end{proof}

    \section{Missing proof}

    In Theorem~\ref{thm:sabIND} we showed using the adversary method that $\QS_\weak(\IND_n) = \Omega(n)$. We now show that the same proof can be adapted to show the same lower bound on $\QS_\strong(\IND_n)$ as well.

    \begin{corollary}
        \label{thm:sabindStrong}
        Let $n$ be a positive integer. Then, $\QS_{\strong}(\IND_n) = \Omega(n)$.
    \end{corollary}
    \begin{proof}
        Let $N = n + 2^n$.
        In the proof of Theorem~\ref{thm:sabIND} we constructed a hard relation $R$ for $f = \IND_{n,\sab}$ to show a $\Omega(n)$ lower bound in the weak sabotage model. Using $R$ we construct a hard relation for $\IND_{n, \sab}^\strong$, which we denote by $R^{\strong}$ in the strong sabotage model, and use Lemma~\ref{lem: ambainis v2}. Define $((x,y,*),(x',y',\dagger)) \in R^{\strong}$ if and only if all of the following hold true:\footnote{Recall that $(x,y,*)$ denotes $((x_j,y_j,z_j))^N_{j=1}$ where $z = \sabStar{x}{y}$ and $(x,y,\dagger)$ denotes $((x_j,y_j,z_j))^N_{j=1}$ where $z = \sabDagger{x}{y}$; see Definition~\ref{def:fsab}.}
        \begin{enumerate}
            \item $(x,y,*)\coloneqq ((x_j,y_j,z_j))^N_{j=1} \in S^{\strong}_{*}$ for $\IND_n$, $(x',y',\dagger)\coloneqq ((x'_j,y'_j,z'_j))^N_{j=1} \in S^{\strong}_{\dagger}$ for $\IND_n$,
            \item $(z,z') \in R$.
        \end{enumerate}
        Following the language of Lemma~\ref{lem: ambainis v2} we have
        \begin{equation}
            X^{\strong}=\{(x,y,*)\coloneqq ((x_j,y_j,z_j))^N_{j=1} : z \in X \text{~as defined in Equation~\ref{eq:XforR}} \}.
        \end{equation}
        Similarly,
        \begin{equation}
            Y^{\strong}=\{(x,y,\dagger)\coloneqq ((x_j,y_j,z_j))^N_{j=1} : z \in Y \text{~as defined in Equation~\ref{eq:YforR}} \}.
        \end{equation}
        We now analyze the quantities $m_{X^{\strong}}$, $m_{Y^{\strong}}$ and $\ell_{\max}$ from Lemma~\ref{lem: ambainis v2}.
        \begin{enumerate}
            \item As described above, the way we construct $R^\strong$ using elements of $R$ ensures that $m_{X^{\strong}} \geq m_X$ and $m_{Y^{\strong}} \geq m_Y$. Additionally, for every $z \in X \cup Y$ there is exactly one pair in $\IND^{-1}(0) \times \IND^{-1}(1)$ for which $z$ is the sabotaged input. Hence, $m_{X^{\strong}} = m_X =\binom{n}{2}$ and $m_{Y^{\strong}} = m_Y = \binom{n}{2}$.
            \item $\ell_{\max}=\max \cbra{\binom{n}{2}, (n-1)^2}$: we consider two cases.
            \begin{enumerate}
                \item $i \in [n]$: Fix an $(x,y,*),(x',y',\dagger) \in R^{\strong}$ differing at an index $i \in [n]$. Recall that $\ell_{(x,y,*),i}$ denotes the number of $(x',y',\dagger) \in Y^{\strong}$ such that $(x,y,*),(x',y',\dagger) \in R^{\strong}$ and $(x_i,y_i,z_i) \neq (x'_i,y'_i,z'_i)$. The construction of sets $X^{\strong}, Y^{\strong}$ (which is in turn based on $X,Y$, respectively) ensures that for all $i \in [n]$ we have $x_i=y_i=z_i$ and $x'_i=y'_i=z'_i$. Hence, directly using the same arguments as in Item~\ref{item:proof}, we get $\ell_{(x,y,*),i}=n-1$ and $\ell_{(x',y',\dagger),i}=n-1$, hence $\ell_{(x,y,*),i} \cdot \ell_{(x',y',\dagger),i}=(n-1)^2$.

                \item $i \in [n+1,N]$: Fix an $(x,y,*)\coloneqq ((x_j,y_j,z_j))^{N}_{j=1},(x',y',\dagger) \coloneqq ((x'_j,y'_j,z'_j))^{N}_{j=1} \in R^{\strong}$ differing at an index $i \in [n+1,N]$. By the structure of $R$ and by extension of $R^{\strong}$, for strings $z\coloneqq(a,b)$ and $ z'\coloneqq (a',b')$ the string $b$ is all-$0$ string except for the $\bin(a)$'th location which is a $*$, and string $b'$ is all-$0$ string except for the $\bin(a')$'th location which is a $\dagger$. Thus the only important case is for $i \in \{n+\bin(a),n+\bin(a')\}$. Without loss of generality, assume $i=n+\bin(a)$, and thus $z_i=b_{\bin(a)}=*$. Moreover, for every $z \in X \cup Y$ there is exactly one pair in $\IND^{-1}(0) \times \IND^{-1}(1)$ that sabotage as $z$. So, we have $\ell_{z,i}=\binom{n}{2}$ and $\ell_{z',i}=1$ as $z'_i=b^{'}_{\bin(a')}=0$. Therefore, $\ell_{z,i} \cdot \ell_{z',i}=\binom{n}{2}$.
            \end{enumerate}
        \end{enumerate}
        Lemma~\ref{lem: ambainis v2} then implies
        \begin{equation}
            \QS_{\strong}(\IND_n)= \Q(\IND_{n, \sab}^\strong) = \Omega \left( \sqrt{\frac{\binom{n}{2}^2}{\max \cbra{\binom{n}{2},(n-1)^2}}} \right)=\Omega(n).
        \end{equation}
    \end{proof}
\end{document}